\documentclass[a4paper,UKenglish,cleveref,autoref]{lipics-v2019}

\newcommand\longversion[1]{}

\title{Maximum Clique in Disk-Like Intersection Graphs}
\titlerunning{Maximum Clique in Disk-Like Intersection Graphs}

\author{\'Edouard Bonnet}{Univ Lyon, CNRS, ENS de Lyon, Universit{\'e} Claude Bernard Lyon 1, LIP UMR5668, France}{edouard.bonnet@ens-lyon.fr}{https://orcid.org/0000-0002-1653-5822}{}

\author{Nicolas Grelier}{Department of Computer Science, ETH Z{\"u}rich}{nicolas.grelier@inf.ethz.ch}{}{Research supported by the Swiss National Science Foundation within the collaborative DACH project Arrangements and Drawings as SNSF Project 200021E-171681.}

\author{Tillmann Miltzow}{Utrecht University, Utrecht Netherlands}{t.miltzow@googlemail.com}{}{}

\authorrunning{\'E. Bonnet, N. Grelier, T. Miltzow}

\Copyright{\'Edouard Bonnet, Nicolas Grelier, Tillmann Miltzow}



\category{}

\relatedversion{}

\supplement{}

\acknowledgements{The third author acknowledges the NWO Veni grant EAGER.}

\nolinenumbers 

\hideLIPIcs  

\EventEditors{John Q. Open and Joan R. Access}
\EventNoEds{2}
\EventLongTitle{42nd Conference on Very Important Topics (CVIT 2016)}
\EventShortTitle{CVIT 2016}
\EventAcronym{CVIT}
\EventYear{2016}
\EventDate{December 24--27, 2016}
\EventLocation{Little Whinging, United Kingdom}
\EventLogo{}
\SeriesVolume{42}
\ArticleNo{23}

\usepackage{xspace}
\usepackage[table]{xcolor}
\usepackage{tikz}
\usepackage{array}
\usetikzlibrary{fit}
\usetikzlibrary{calc}
\usetikzlibrary{shapes}
\usetikzlibrary{decorations.pathreplacing}

\tikzstyle{vp}=[circle,fill,inner sep=0pt, minimum size=0.1cm]
\tikzstyle{vps}=[circle,fill,inner sep=0pt, minimum size=0.065cm]

\usepackage{relsize}
\usepackage{amsmath,amssymb,amsthm}
\usepackage{caption}
\usepackage{subcaption}
\usepackage{flexisym}
\usepackage{framed}
\usepackage{float}
\usepackage{graphicx}
\usepackage{verbatim}

\makeatletter
\newcommand\footnoteref[1]{\protected@xdef\@thefnmark{\ref{#1}}\@footnotemark}
\makeatother

\newcommand{\cli}{\textsc{Maximum Clique}\xspace}

\newcommand{\tsat}{\textsc{3-SAT}\xspace}
\newcommand{\tsatb}{\textsc{3-SAT-B}\xspace}
\newcommand{\ksat}{\textsc{$k$-SAT}\xspace}
\newcommand{\ksatb}{\textsc{$k$-SAT-B}\xspace}

\newcommand{\nae}{\textsc{Not-All-Equal 3-SAT}\xspace}
\newcommand{\naeb}{\textsc{Not-All-Equal 3-SAT-B}\xspace}

\newcommand{\snaeb}{\textsc{NAE 3-SAT-B}\xspace}

\newcommand{\snaet}{\textsc{NAE 3-SAT-3}\xspace}

\newcommand{\pnaet}{\textsc{Positive Not-All-Equal 3-SAT-3}\xspace}
\newcommand{\psnaet}{\textsc{Positive NAE 3-SAT-3}\xspace}

\newcommand{\fnae}{\textsc{Not-All-Equal 4-SAT}\xspace}
\newcommand{\fnaeb}{\textsc{Not-All-Equal 4-SAT-B}\xspace}
\newcommand{\fsnaeb}{\textsc{NAE 4-SAT-B}\xspace}

\newcommand{\knae}{\textsc{Not-All-Equal $k$-SAT}\xspace}
\newcommand{\knaeb}{\textsc{Not-All-Equal $k$-SAT-B}\xspace}
\newcommand{\ksnae}{\textsc{NAE $k$-SAT}\xspace}

\newcommand{\mipa}{\textsc{Max Interval Permutation Avoidance}\xspace}
\newcommand{\smipa}{\textsc{MIPA}\xspace}

\newcommand{\mis}{\textsc{MIS}\xspace}

\newcommand{\midd}{\text{middle}}
\theoremstyle{plain}

\newcommand{\defproblem}[3]{
 \vspace{1mm}
\noindent\fbox{
 \begin{minipage}{0.96\textwidth}
 \begin{tabular*}{\textwidth}{@{\extracolsep{\fill}}lr} #1 & \\ \end{tabular*}
 {\bf{Input:}} #2 \\
 {\bf{Goal:}} #3
 \end{minipage}
 }
 \vspace{1mm}
}

\begin{document}

\maketitle

\begin{abstract}
  We study the complexity of \textsc{Maximum Clique} in intersection graphs of convex objects in the plane.
  On the algorithmic side, we extend the polynomial-time algorithm for unit disks [Clark '90, Raghavan and Spinrad '03] to translates of any fixed convex set.
  We also generalize the efficient polynomial-time approximation scheme (EPTAS) and subexponential algorithm for disks [Bonnet et al. '18, Bonamy et al. '18] to homothets of a fixed centrally symmetric convex set.
  
  The main open question on that topic is the complexity of \textsc{Maximum Clique} in disk graphs.
  It is not known whether this problem is NP-hard.
  We observe that, so far, all the hardness proofs for \textsc{Maximum Clique} in intersection graph classes $\mathcal I$ follow the same road.
  They show that, for every graph $G$ of a large-enough class $\mathcal C$, the complement of an even subdivision of $G$ belongs to the intersection class $\mathcal I$.
  Then they conclude invoking the hardness of \textsc{Maximum Independent Set} on the class $\mathcal C$, and the fact that the even subdivision preserves that hardness.
  However there is a strong evidence that this approach cannot work for disk graphs [Bonnet et al. '18].
  We suggest a new approach, based on a problem 
  that we dub \textsc{Max Interval Permutation Avoidance}, 
  which we prove unlikely to have a subexponential-time approximation scheme.
  We transfer that hardness to \textsc{Maximum Clique} in intersection graphs of objects which can be either half-planes (or unit disks) or axis-parallel rectangles.
  That problem is not amenable to the previous approach.
  We hope that a scaled down (merely NP-hard) variant of 
  \textsc{Max Interval Permutation Avoidance} 
  could help making progress on the disk case, 
  for instance by showing the NP-hardness for (convex) pseudo-disks.
\end{abstract}

\section{Introduction}

In an \emph{intersection graph}, the vertices are represented by sets and there is an edge between two sets whenever they intersect.
Of course if the sets are not restricted, every graph is an intersection graph.
Interesting proper classes of intersection graphs are obtained by forcing the sets to be some specific geometric objects.
This comprises unit interval, interval, multiple-interval, chordal, unit disk, disk, axis-parallel rectangle, segment, and string graphs, to name a few.
For the most part, they transparently consist of all the intersection graphs of the corresponding objects.
Note that \emph{strings} are (polygonal) curves in the plane, and that chordal graphs are the intersection graphs of subtrees in a tree.
Intersection graphs have given rise to books (see for instance \cite{McKee99}, where applications to biology, psychology, and statistics, are detailed) chapters in monographs (as in \cite{Brandstadt99}), surveys \cite{Fishkin03,Hlineny01}, and theses \cite{EJvL2009}. 
In this paper we consider objects that are convex sets in the plane.

One especially interesting problem on geometric intersection graphs is \cli.
The first reason is that \cli is neither a packing 
nor a covering problem, for which our theoretical understanding is rather comprehensive.  
Packing problems (such as \textsc{Maximum Independent Set}) and covering problems 
(such as \textsc{Dominating Set}) are often NP-hard in intersection graphs since these 
problems are already hard on planar graphs.
Note for instance that disk intersection graphs~\cite{Koebe36} and segment
 intersection graphs~\cite{Chalopin09} both contain all the planar graphs.
It turns out that \textsc{Maximum Independent Set} (\mis) 
and \textsc{Dominating Set} remain intractable in unit disk, 
unit square, or segment intersection graphs \cite{Marx06}: 
Not only they are NP-hard but, being W[1]-hard, they are
 unlikely to admit a fixed-parameter tractable (FPT), 
 that is, $f(k)n^{O(1)}$-time algorithms, with $n$ being the 
 input size, $k$ the size of the solution, and $f$ any computable function. 
This intractability is sharply complemented by PTASes for many problems \cite{Chan03,Nieberg04,Nieberg05,Erlebach05,Leeuwen06,EJvL2009}, 
whereas efficient PTASes (EPTASes) are ruled out by the W[1]-hardness of Marx~\cite{Marx06}.
The existence or unlikelihood of subexponential algorithms 
for various problems on segment and string graphs was conducted in \cite{BonnetR19}.

On the contrary, many questions are still open when it comes 
to the computational complexity of \cli in intersection graphs.
Clark et al. \cite{Clark90} show a polynomial-time algorithm for unit disks.
A randomized EPTAS, deterministic PTAS, and 
subexponential-time algorithm were recently obtained for general disk graphs~\cite{Bonnet18,Bonamy18}.
However neither a polytime algorithm nor 
NP-hardness is currently known for \cli on disk graphs.
Making progress on this open question is 
the main motivation of the paper.
\cli was shown NP-hard in segment intersection 
graphs by Cabello et al.~\cite{CabelloCL13}.
The proof actually carries over to intersection 
graphs of unit segments or rays.
The existence of an FPT algorithm or of a 
subexponential-time algorithm for \cli in segment graphs are both open.
\cli can be solved in polynomial-time in axis-parallel 
rectangle intersection graphs, since their number of 
maximal cliques is at most quadratic 
(every maximal clique corresponds to a distinct cell in any representation).
This result was generalized to $d$-dimensional polytopes 
whose facets are all parallel to $k$ fixed $(d-1)$-dimensional hyperplanes, 
where \cli can be solved in time $n^{O(d k^{d+1})}$~\cite{Brimkov18}.
Note that if the rectangles may have arbitrary slopes, then \cli is NP-hard since the class then contains segment graphs. 

The second reason, to study \cli, is that it translates into a very natural question: what is the maximum subset of pairwise intersecting objects?
For unit disks, this is equivalent to looking for the maximum subset of centers with (geometric) diameter~2.
This is a useful primitive in the context of clustering a given set of points.
A related question with a long history is the number of points necessary and sufficient to pierce a collection of pairwise intersecting disks.
Danzer \cite{Danzer86} and Stacho \cite{Stacho81} independently showed that four points are sufficient and sometimes necessary.
Recently Har-Peled et al. \cite{HarPeled18} gave a linear-time algorithm to find five points piercing a pairwise intersecting collection of disks.
A bit later, Carmi et al. \cite{Carmi18} obtained a linear-time algorithm for only four points.

Up to this point, we remained vague on how the input intersection graph was given.
For, say, disk graphs, do we receive the mere abstract graph or a list of the disks specified by their centers and radii?
Computing the graph from the geometric representation can be done efficiently, but not the other way around.
Indeed recognizing disk graphs is NP-hard~\cite{Breu98} and even $\exists \mathbb R$-complete~\cite{Kang12}, where $\exists \mathbb R$ is a class between NP and PSPACE of all the problems polytime reducible to solving polynomial inequalities over the reals.
Recognizing string graphs is NP-hard~\cite{Kratochvil91}, and rather unexpectedly in NP~\cite{Schaefer03}, while recognizing segment graphs is $\exists \mathbb R$-complete~\cite{Kratochvil94}.
In this context, an algorithm is said \emph{robust} if it does not require the geometric representation.
A polytime robust algorithm usually decides the problem for a \emph{proper superclass} of the intersection graph class at hand, or correctly reports that the input does not belong to the class.
Hence the robust algorithm does not imply an efficient recognition of the class.
The polynomial-time algorithm of Clark et al.~\cite{Clark90} for \cli in unit disk intersection graphs requires the geometric representation.
Raghavan and Spinrad later extended it to an efficient robust algorithm~\cite{Raghavan03}.

\paragraph*{A new alternative to the \emph{co-2-subdivision} approach}

\mis, which boils down to \cli on the complement graphs, is APX-hard on subcubic graphs \cite{Alimonti00}.
A folklore self-reduction first discovered by Poljak~\cite{Poljak74} consists of subdividing each edge of the input graph twice (or any even number of times).
One can show that this reduction preserves the APX-hardness.
Therefore, a way to establish such an intractability for \cli on a given intersection graph class is to show that for every (subcubic) graph $G$, its complement of 2-subdivision $\overline{\text{Subd}_2(G)}$ (or $\overline{\text{Subd}_s(G)}$ for a larger even integer $s$, see~\cite{Francis15}) is representable.
\mis admits a PTAS on planar graphs, but remains NP-hard.
Hence showing that for every (subcubic) planar graph $G$, the complement of an even subdivision of $G$ is representable shows the simple NP-hardness (see~\cite{CabelloCL13,Francis15}).

So far, representing complements of even subdivisions of graphs belonging to a class on which \mis is NP-hard (resp. APX-hard) has been the main, if not unique\footnote{Admittedly Butman et al.~\cite{Butman10} showed that \cli is NP-hard on 3-interval graphs, by reducing from \textsc{Max 2-DNF-SAT} which is very close to \textsc{Max Cut}. However this result was later subsumed by~\cite{Francis15}.}, approach to show the NP-hardness (resp. APX-hardness) of \cli in geometric intersection graph classes.  
This approach was used by Middendorf and Pfeiffer \cite{Middendorf92} for some restriction of string graphs, the so-called \emph{$B_1$-VPG graphs}, by Cabello et al. \cite{CabelloCL13} to settle the then long-standing open question of the complexity of \cli for segments (with the class of planar graphs), by Francis et al.~\cite{Francis15} for 2-interval, unit 3-interval, 3-track, and unit 4-track graphs (with the class of all graphs; showing APX-hardness), and unit 2-interval and unit 3-track graphs (with the class of planar graphs; showing only NP-hardness), by Bonnet et al.~\cite{Bonnet18} for filled ellipses and filled triangles, and by Bonamy et al.~\cite{Bonamy18} for ball graphs, and 4-dimensional unit ball graphs.
Bonnet et al.~\cite{Bonnet18} show that the complement of two mutually induced odd cycles is not a disk graph.
As a consequence, to show the NP-hardness of \cli on disk graphs with the described approach, one can only hope to represent all the graphs without two mutually induced odd cycles.
However we do not know if \mis is even NP-hard in that class.

The main conceptual contribution of the paper is to suggest an alternative to that approach.
We introduce a technical intermediate problem that we call \mipa (\smipa, for short), which is a convenient way of seeing \textsc{Max Cut}.
We prove that \smipa is unlikely to have an approximation scheme running in subexponential time.
We then transfer that lower bound to \cli in the intersection graphs of objects that can be either unit disks or axis-parallel rectangles;
a class for which the \emph{co-2-subdivision} approach does not seem to work. 
Recall that when all the objects are unit disks or when all the objects are axis-parallel rectangles, polynomial-time algorithms are known.

\begin{figure}[h!]
  \centering
  \begin{subfigure}[b]{0.45\textwidth}
    \centering
    \begin{tikzpicture}
      \def\n{8}
      \def\m{6}
      \def\s{0.6}
      \foreach \i in {1,...,\n}{
        \node[draw,circle] (e\i1) at (\i * \s,0) {} ;
        \node[draw,circle] (e\i2) at (\i * \s,1) {} ;
        \draw[dashed] (e\i1) -- (e\i2) ;
      }
      \node[draw, rectangle, rounded corners, fit=(e11) (e\n2)] (e) {} ;
      \node at (0,0.5) {$E$} ;
      \foreach \i in {1,...,\m}{
        \node[draw,circle] (v\i) at (\i * \s + \s,-1) {} ;
      }
      \node[draw, rectangle, rounded corners, fit=(v1) (v\m)] (v) {} ;
      \node at (0.5,-1) {$V$} ;
      \foreach \i/\j in {2/1,3/1,5/2}{
        \draw[dashed] (v3) -- (e\i\j) ;
      }
      \foreach \i/\j in {2/1,3/1,5/2}{
        \node[circle,fill=red,opacity=0.7,inner sep=0.1cm] at (\i * \s,\j - 1) {} ;
      }
      \foreach \i/\c in {1/yellow,2/orange,3/red,4/green,5/blue,6/purple}{
        \node[circle,fill=\c,opacity=0.7,inner sep=0.1cm] at (\i * \s + \s,-1) {} ;
      }
    \end{tikzpicture}
    \caption{Co-2-subdivision of subcubic graphs: edges are represented by an antimatching, vertices, by a clique.}
    \label{fig:co-2-subd-approach}
  \end{subfigure}
  \qquad
  \begin{subfigure}[b]{0.45\textwidth}
    \centering
    \begin{tikzpicture}
      \def\n{6}
      \def\m{8}
      \def\s{0.6}
      \foreach \i in {1,...,\n}{
        \node[draw,circle] (v\i1) at (\i * \s,0) {} ;
        \node[draw,circle] (v\i2) at (\i * \s,1) {} ;
        \draw[dashed] (v\i1) -- (v\i2) ;
      }
      \node[draw, rectangle, rounded corners, fit=(v11) (v\n2)] (v) {} ;
      \node at (0,0.5) {$V$} ;
      \foreach \i in {1,...,\m}{
        \node[draw,circle] (e\i) at (\i * \s - \s,-1) {} ;
      }
      \node[draw, rectangle, rounded corners, fit=(e1) (e\m)] (e) {} ;
      \node at (-0.75,-1) {$E$} ;
      \draw[dashed] (v21) -- (e5) -- (v52) ;
      \foreach \i/\c in {1/yellow,2/orange,3/red,4/green,5/blue,6/purple}{
        \node[circle,fill=\c,opacity=0.7,inner sep=0.1cm] at (\i * \s,0) {} ;
        \node[circle,fill=\c,opacity=0.7,inner sep=0.1cm] at (\i * \s,1) {} ;
      }
      \fill[orange,opacity=0.7] (4 * \s, -0.86) arc (90:270:0.9ex);
      \fill[blue,opacity=0.7] (4 * \s, -1.133) arc (-90:90:0.9ex);
    \end{tikzpicture}
    \caption{\smipa approach: vertices are represented by an antimatching with constant weight, edges, by a clique.}
    \label{fig:mipa-approach}
  \end{subfigure}
  \caption{Dashed segments represent non-edges.
    Both the \emph{co-2-subdivision} and the MIPA approaches require to construct an antimatching and a clique.
    In the \emph{co-2-subdivision} approach, the \emph{clique vertices} have co-degree $3$ to the antimatching.
    In the MIPA approach their co-degree is only $2$.
    While the difference is seemingly small, the graph class formed by axis-parallel rectangles and unit disks is not amenable to the \emph{co-2-subdivision} approach (see \cref{sec:disks-and-rectangles}).}
  \label{fig:comparison}
\end{figure}
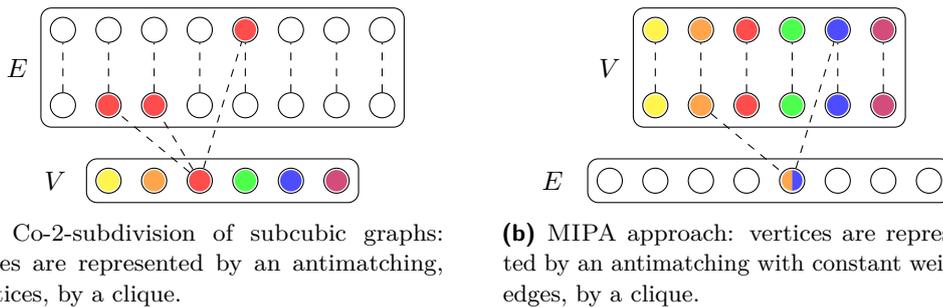

Anticipating on \cref{sec:disks-and-rectangles} where \smipa is defined, one can already see on \cref{fig:comparison} that both approaches require to represent an antimatching (i.e., a complement of an induced matching), a clique, and some relation between them.
Antimatchings (and obviously cliques) of arbitrary size are representable by half-planes and unit disks.
The difficulty in both cases is to get the right adjacencies between the antimatching and the clique.
The \smipa approach only needs the vertices of the clique to \emph{avoid} two vertices in the antimatching, whereas this number is at least three in the \emph{co-2-subdivision} approach.
This seemingly small difference is actually crucial, as we will see in \cref{sec:disks-and-rectangles}.

A transparent reduction from \smipa yields APX-hardness, which is ruled out for disks, due to the EPTAS.
To be usable for disk graphs, one will thus need to scale \smipa down\footnote{for instance, by replacing the arbitrary matching $M$ by \emph{pseudo-disk-like} objects} to a NP-hard problem admitting a PTAS.
In doing so, one should keep in mind that \textsc{Planar Max Cut} \cite{Dorfman72} and \textsc{Planar Not-All-Equal SAT} \cite{Moret88} are solvable in polynomial-time.
A next step could be to show the NP-hardness of \cli for (convex) pseudo-disks.
It turns out to be already quite delicate.
There is a distinct possibility that convex pseudo-disks have constant induced odd cycle packing number (see \cref{sec:prelim} for the definitions).
This would imply a subexponential-time algorithm and an EPTAS \cite{Bonnet18,Bonamy18}, and that one would need a scaled down version of \smipa to establish NP-hardness even in that case.  

\paragraph*{Organization of the paper}

The rest of the paper is organized as follows.
In \cref{sec:prelim}, we introduce the relevant set, graph, geometry notations and definitions.
Then we give the necessary background in hardness of approximation to get ready for the next section.   
In \cref{sec:disks-and-rectangles}, we introduce the \mipa problem and prove that it is unlikely to admit a subexponential-time approximation scheme.
We use it to show that adding half-planes or unit disks to axis-parallel rectangles, is enough for \cli to go from trivially in P to APX-hard.
This is a proof of concept for a different road-map to the \emph{co-even-subdivision} approach, which we know cannot work for disk graphs.
We also observe that if the half-planes are not allowed to be parallel (hence pairwise intersect), then the problem becomes tractable.
In \cref{sec:translates}, we extend the polytime algorithm for unit disks \cite{Clark90,Raghavan03} to translates of a fixed convex set.
In \cref{sec:homothets}, we extend the EPTAS for disks \cite{Bonnet18,Bonamy18} to homothets of a fixed convex centrally symmetric set.
Our algorithms are robust and our lower bound also holds when the geometric representation is given.

\section{Preliminaries}\label{sec:prelim}

\paragraph*{Sets and graphs}

For a pair of positive integers $i \leqslant j$, $[i,j]$ denotes the set of all the integers that are at least $i$ and at most $j$, and $[i]$ is a short-hand for $[1,i]$.
We overload the notation $[\cdot,\cdot]$: If it is explicit or clear from the context that $x$ or $y$ is non-integral, then $[x,y]$ denotes the set of reals that are at least $x$ and at most $y$.
We use the usual notations and definitions of graph theory, as they can be found for example in Diestel's book~\cite{Diestel12}.
We denote by $K_t$, $C_t$, and $K_{s,t}$ the complete graph (or clique) on $t$ vertices, the cycle on $t$ vertices, and the biclique on $s$ and $t$ vertices.
The graph $\overline G$ denotes the \emph{complement} of $G$, obtained by flipping edges into non-edges, and non-edges into edges.  
Subdividing an edge $e=uv$ consists of adding a new vertex linked to both $u$ and $v$, and removing the edge $e$.
The 2-subdivision $\text{Subd}_2(G)$ of a graph $G$ is obtained by subdividing each of its edges twice.
An even subdivision of a graph $G$ consists of subdividing every edge of $G$ an even number of times (potentially zero).
A cycle is said \emph{induced} if it is chordless.
An \emph{odd cycle} (resp.~\emph{even cycle}) is a cycle on an odd (resp.~even) number of vertices.
One can observe that an odd cycle always contains an induced odd cycle.
Two cycles are said \emph{mutually induced} if they are chordless and there is no edge linking a vertex of one to a vertex of the other.
The induced odd cycle packing number is the maximum number of disjoint odd cycles, that are pairwise mutually induced.
Here an \emph{antimatching} is the complement of an \emph{induced matching} (i.e., a disjoint union of edges).
We say that a graph $G$ is \emph{representable} by some geometric objects, if translates of these objects may have $G$ as intersection graph. 

\paragraph*{Geometric notations}

In this paper, we only consider sets in the plane.
For two distinct points $a$ and $b$, $\ell(a,b)$ denotes the line going through $a$ and $b$.
A set $S$ is \emph{convex} if for any two distinct points $a$ and $b$ in $S$, the line segment with endpoints $a$ and $b$ is contained in $S$.
It is \emph{bounded}, if it is contained in some disk.
A set $S$ is said to be \emph{centrally symmetric} about the origin if for any point $a$ in $S$, $-a$ is also in $S$.
We mostly deal with sets that are bounded, centrally symmetric, and convex, as they are a natural generalization of disks. 

For two sets $S_1$ and $S_2$, we denote by $S_1+S_2:=\{s_1+s_2\mid s_1 \in S_1, s_2 \in S_2\}$ their Minkowski sum.
For the sake of simplicity, for any point $c$ and any set $S$, we denote by $c+S$ the Minkowski sum of $\{c\}$ and $S$.
$S'$ is a \emph{translate} of $S$ if there exists $c$ such that $S'=c+S$.
Given a positive real number $\lambda$, $\lambda S$ denotes the set $\{\lambda s\mid s \in S\}$.
We say that $S'$ is a \emph{homothet} of $S$ if there exist a positive real number $\lambda$ and a point $c$ such that $S'=c+\lambda S$.
Moreover we name $c$ the \emph{center} of $S'$, and $\lambda$ its \emph{scaling factor}. 

Let $F$ be a family of sets in the plane.
They form a \emph{pseudo-disk arrangement} if for any pair of sets of $F$, their boundaries intersect at most twice.
If the sets are also convex we refer to them as \emph{convex pseudo-disks}.
They also constitute a natural generalization of disks.
Rectangles are \emph{axis-parallel} if their boundaries have only two different slopes.
A rectangle is an $\varepsilon$-square if its length divided by its width is smaller than $1+\varepsilon$.

\paragraph*{Approximation-schemes}

A \emph{polynomial-time approximation scheme} (PTAS) for a maximization problem is an algorithm which takes, together with its input, a parameter $\varepsilon > 0$ and outputs in time $n^{f(\varepsilon)}$ a solution of value at least $(1-\varepsilon)\text{OPT}$, where $\text{OPT}$ is the optimum value.
An \emph{efficient} PTAS (EPTAS) is the same but has running time $f(\varepsilon)n^{O(1)}$.
Note that the existence of an EPTAS, for a problem in which the objective value is the size of the solution $k$, implies an FPT algorithm in $k$, by setting $\varepsilon$ to $1-\frac{1}{k+1}$.
Indeed in time $f(1-\frac{1}{k+1})n^{O(1)}=g(k)n^{O(1)}$, one then obtains an \emph{exact} solution.
A \emph{quasi} PTAS (QPTAS) is an approximation scheme with running time $n^{\text{polylog}~n}$, for every $\varepsilon > 0$.
Less standardly, we call \emph{subexponential} AS (SUBEXPAS) an approximation scheme with running time $2^{n^\gamma}$ for some $\gamma < 1$, for every $\varepsilon > 0$.
These approximation schemes can come deterministic or randomized.
A maximization problem $\Pi$ is \emph{APX-hard} if there is a constant $\gamma < 1$ such that $\gamma$-approximating $\Pi$ is NP-hard.
Unless P$=$NP, an APX-hard problem cannot admit a PTAS.
Ruling out a SUBEXPAS (under admittedly a stronger assumption than P$\neq$NP) constitutes a sharper inapproximability than the APX-hardness.

\paragraph*{The Exponential-Time Hypothesis and Probabilistically Checkable Proofs}

This section contains some folklore results and their necessary background leading to the inapproximability of \pnaet, and more precisely that a SUBEXPAS for that problem is unlikely.
The proofs are given for the sake of self-containment.
To our knowledge, such a strong inapproximability is \emph{known} for \pnaet or for the closely related \textsc{Max Cut} (see for instance~\cite{Bonnet15}) but was never fully written up, nor the constants were worked out.
For a reader eager to discover a simple but powerful idea, the highlight is perhaps the use of an expander graph to encode a ``global variable set to false'', while keeping constant the number of occurrences per variable.
This is a topical idea in hardness of approximation for graph problems with bounded degree or satisfiability problems with bounded variable occurrence.

The Exponential-Time Hypothesis (ETH, for short) of Impagliazzo and Paturi~\cite{ImpagliazzoETH} asserts that there is no subexponential-time algorithm solving \ksat.
More precisely, for every integer $k \geqslant 3$, there is an $\varepsilon > 0$ such that \ksat cannot be solved in time $2^{\varepsilon n}$ on $n$-variable instances.
If we define $s_3$ (taking the same notation as in the original paper) as the infimum of the reals $\delta$ such that \tsat can be solved in time $2^{\delta n}$, then the ETH can be expressed as $s_3 > 0$. 
Impagliazzo et al.~\cite{ImpagliazzoSparsification} present a subexponential-time Turing-reduction parameterized by a positive real $\varepsilon > 0$ which, given a \ksat-instance $\phi$ with $n$ variables and $m$ clauses, produces at most $2^{\varepsilon n}$ \ksat-instances $\phi_1, \ldots, \phi_t$ such that $\phi \Leftrightarrow \bigvee_{i \in [t]} \phi_i$, each $\phi_i$ having no more than $n$ variables and $C_\varepsilon n$ clauses for some constant $C_\varepsilon$ (depending solely on $\varepsilon$, and \emph{not} on $n$ and $m$).
This important reduction is known as the Sparsification Lemma.
One can observe that, due to the Sparsification Lemma, there is an $\varepsilon > 0$ such that there is no algorithm solving \ksat in time $2^{\varepsilon m}$ on $m$-clause instances, assuming that the ETH holds.
For the sparsification of a \tsat-instance, the constant $C_\varepsilon$ can be upper-bounded by $10^8(1/\varepsilon)^2 \log^2(1/\varepsilon)$.
One can sparsify a \tsat-instance in $2^{\frac{s_3}{2} n}$ instances with at most $B := C_{s_3/2} \leqslant 10^8(2/s_3)^2 \log^2(2/s_3)$ occurrences per variable.
Assuming the ETH, these sparse instances cannot be solved in time $2^{\frac{s_3}{2} n}$.

Probabilistically Checkable Proofs (PCPs, for short) have laid the foundation of the hardness of approximation, providing the first non-trivial examples of NP-hard so-called \emph{gap} problems.
In a PCP with perfect completeness PCP$_{1,\varepsilon}(r,q)$, a randomized polytime verifier, using $r$ random bits and making $q$ bit queries to a proof $\pi$, tries to decide if an input $x$ is positive (in the language) or negative (outside the language).
The verifier should always accept any positive instance $x$ being given a correct proof $\pi$, and, for every (wrong) proof, should reject a negative instance $x$ with probability at least $1-\varepsilon$.
Overall the verifier cannot query more than $2^r q$ positions in the proof, so we can assume that the proof has length at most $2^r q$.
Moshkovitz and Raz~\cite{MoshkovitzR10} built a PCP$_{1,\varepsilon}(r=(1+o(1)) \log n+\log(1/\varepsilon),q=2)$ over an alphabet of size $2^{\text{poly}(1/\varepsilon)}$ to decide $n$-variable \tsat-instances, for any $\varepsilon > 0$ even function of~$n$.
Fixing $1/\varepsilon$ to be polylogarithmic in $n$, this gives proof size $n^{1+o(1)}$, as well as error $\varepsilon = o(1)$. 
Combining the Sparsification Lemma~\cite{ImpagliazzoSparsification} (applied first since it does not preserve inapproximability) with the polytime inapproximability result of H{\aa}stad~\cite{Hastad01}, improved to subexponential-time inapproximability by the PCP of Moshkovitz and Raz~\cite{MoshkovitzR10}, one obtains the following:
\begin{theorem}\cite{Hastad01,MoshkovitzR10,ImpagliazzoSparsification}\label{thm:hardness-tsat-B}
  Under the ETH, for every $\delta > 0$ one cannot distinguish in time $2^{n^{1-\delta}}$, $n$-variable $m$-clause \tsat-instances that are satisfiable from instances where at most $(7/8+o(1))m$ clauses can be satisfied, even when each variable appears in at most $B$ clauses.
  Thus \tsatb cannot be $7/8+o(1)$-approximated in time $2^{n^{1-\delta}}$.
\end{theorem}

In \cref{thm:hardness-tsat-B}, \tsatb stands for the \tsat-problem with the additional guarantee that every instance have at most $B$ occurrences of each of its variables.
In \cref{sec:disks-and-rectangles} we will need such an inapproximability result for the \nae-problem with a bounded number of occurrences per variable.
We recall the definition of \knae (\ksnae, for short).

\defproblem{\knae}{A conjunction of $m$ ``clauses'' $\phi = \bigwedge_{i \in [m]} C_i$ each on at most $k$ literals.}{Find a truth assignment of the $n$ variables such that each ``clause'' has at least one satisfied literal and at least one non-satisfied literal.}

The \knaeb-problem is the same but each variable appears in at most $B$ clauses (similarly as for \ksatb).
The adjective \textsc{Positive} prepended to a satisfiability problem means that no negation (or \emph{negative literal}) can appear in its instances.
As a slight abuse of notation, we keep the same problem names for the maximization versions, where all the clauses may not be simultaneously satisfied but the goal is to satisfy the largest fraction of them.
We will mostly deal with the maximization versions, and this should be clear from the context.
Another abuse of notation is that we call \emph{clauses} the \emph{not-all-equal} constraints, and still denote them with $\lor$. 
The performance guarantee of an approximation algorithm is then defined as the minimum of \emph{$\text{number of satisfied clauses}/m$} taken over all the instances.

There is a folklore reduction from \tsat to \nae which transfers the APX-hardness of the former problem to the latter.
This reduction introduces a variable (representing the value \emph{false}) in all the clauses.
To get rid of that variable with many occurrences, we replace it by a network of constraints forming an expander: there is a fresh variable per node of the expander, and an equality constraint linking every adjacent nodes.
That way the number of constraints remain linear, and if a sizable fraction of its "occurrences" are set to true and a sizable fraction of its "occurrences" are set to false, then many clauses are not satisfied in the cut that they induce.

The edge expansion $h(G)$ of a graph $G$ is defined as $$h(G) := \min\limits_{0 < |S| \leqslant |V(G)|/2} \frac{|\partial S|}{|S|}$$
where $\partial S$ is the set of edges with one endpoint in $S$ and the other endpoint outside $S$.
A foundational inequality in the theory of expanders relates the edge expansion $h(G)$ and the second-largest eigenvalue $\lambda_2(G)$ (of the adjacency matrix) of any $d$-regular graph $G$: $h(G) \geqslant \frac{1}{2}(d-\lambda_2(G))$ (see \cite{Dodziuk84,Hoory06,Alon16}).
This is sometimes called Cheeger's inequality.
Gabber and Galil \cite{Gabber81} showed that the following deterministic construction, due to Margulis, admits a relatively large gap between the largest eigenvalue $d=8$, and the second-largest eigenvalue upper-bounded by $5 \sqrt 2 < 7.08$.

\begin{theorem}\cite{Gabber81,Hoory06}\label{thm:expander}
  For every sufficient large natural $n$, the 8-regular graph $H := H(n^2,8)$ on vertex-set $\mathbb Z_n \times \mathbb Z_n$ where every vertex $(x, y)$ is adjacent to the eight vertices $(x \pm 2y,y), (x \pm (2y+1),y), (x,y \pm 2x), (x,y \pm (2x+1))$ satisfies $\lambda_2(H) \leqslant 5 \sqrt 2$.
  Hence, by Cheeger's inequality, $h(H) \geqslant \frac{1}{2}(8 - 5 \sqrt 2) > 0.46$.
\end{theorem}
We observe that $H$ can be easily computed in deterministic polytime, and that for every non-trivial proper subset $S$ of $V(H)$, the number of edges in the cut $(S,V(H) \setminus S)$ is at least  $0.46 \cdot \min(|S|,|V(H) \setminus S|)$.  

\begin{theorem}\label{thm:hardness-fnae-B}
  Under the ETH, for every $\delta > 0$ one cannot distinguish in time $2^{n^{1-\delta}}$, $n$-variable $m$-clause \fnae-instances that are satisfiable from instances where at most $4991 m/5000$ clauses can be satisfied, even when each variable appears in at most $B$ clauses.
  Thus \fnaeb cannot be $4991/5000$-approximated in time $2^{n^{1-\delta}}$.
\end{theorem}

\begin{proof}
  Let $\psi$ be any instance of \tsatb with $N$ variables and $M \leqslant BN$ clauses.
  We start by padding $\psi$ with dummy clauses $d_i \lor \neg d_i$ with fresh variables $d_i$ with only two occurrences overall (both in their dummy clause) until the number of clauses at least doubles and reaches a perfect square.
  This can be done in a way that the number of clauses at most triples, and the number of variables is multiplied by a constant factor $(1+2B)B$.
  Let $n \leqslant N+2M \leqslant (1+2B)N$ and $m \leqslant 3M \leqslant 3BN$ be the number of variables and clauses of the new equivalent \tsatb-instance $\phi := \bigvee_{j \in [m]} C_j$.

  We now describe a linear reduction $\rho$ from \tsatb to \fnaeb starting from the padded instance $\phi$.
  For every clause $C_j$, we introduce two clauses in the \fsnaeb-instance $\rho(\phi)$: $C_j \lor \neg z_j$ and $\overline{C_j} \lor z_j$, where $z_j$ is a fresh variable and $\overline{C_j}$ is obtained from $C_j$ by switching the sign of every literal therein.
  Let $H$ be an 8-regular expander graph provided by \cref{thm:expander} on the vertex set $[m]$.
  (This is where we needed that $m$ is a perfect square.)
  For every edge $ab \in E(H)$ with $a<b$, we add a clause $z_a \lor \neg z_b$. 
  This finishes the reduction and produces instances with $n' := n+m \leqslant (B+1)n \leqslant (B+1)(1+2B)N$ variables and $m' := m+4m = 5m \leqslant 15M \leqslant 5Bn \leqslant 5B(1+2B)N$ clauses.
  Note that each variable appears at most $B$ times, and that the size of the clauses is at most four.
  Since $B \geqslant 9$, the number of occurrences of the new variables $z_j$ does not exceed this threshold.
  
  If $\psi$ is satisfiable, then the same truth assignment augmented by setting all the variables $z_j$ to true (and all the dummy variables $d_i$ indifferently) is a satisfying assignment for the \fsnaeb-instance $\rho(\phi)$.
  Every clause $z_a \lor \neg z_b$ (with $ab \in E(H)$) is satisfied since $z_a$ and $z_b$ are both set to true.
  Every clause $C_j \lor \neg z_j$ is satisfied since at least one literal of $C_j$ is set to true, and the literal $\neg z_j$ is false.
  Symmetrically $\overline{C_j} \lor z_j$ is satisfied since at least one literal of $\overline{C_j}$ is false (namely, the opposite literal satisfying $C_j$), and $z_j$ is true.

  We now assume that at most $(7/8+o(1))M$ clauses of $\psi$ can be satisfied, and we wish to upper-bound the number of satisfiable clauses in $\rho(\phi)$.
  Since we padded $\psi$ with at most $2M$ dummy clause to create $\phi$, at most $(23/24+o(1))m$ clauses of $\phi$ are satisfiable.
  Let $\mathcal V'$ be any assignment of the variables of $\rho(\phi)$, and $\mathcal V$ its restriction to the variables of $\phi$.
  By assumption $\mathcal V$ leaves unsatisfied at least $(1/24-o(1))m$ of the clauses $C_j$ of $\phi$.
  Let us denote these clauses by $C_{u_1}, \ldots, C_{u_t}$ with $t \geqslant (1/24-o(1))m$.
  Note that $t \leqslant m/2$ since we at least doubled $\psi$ with dummy clauses which cannot be unsatisfied. 
  Let $\mathcal T \subseteq [m]$ the indices of the variables $z_j$ set to true, and $\mathcal F := [m] \setminus \mathcal T$.
  Note that all the indices $u_h$ ($h \in [t]$) that are in $\mathcal T$ correspond to clauses $C_j \lor \neg z_j$ (thus $\overline{C_j} \lor z_j$) of $\rho(\phi)$ that are not satisfied.
  Either (case 1) at least half of the $u_h$ are in $\mathcal T$, and then $\mathcal V'$ leaves $2(1/48-o(1))m=(1/24-o(1))m$ clauses of $\rho(\phi)$ unsatisfied.
  Or (case 2) at least half of the $u_h$ are in $\mathcal F$.
  In that case, $|\mathcal F| \geqslant (1/48-o(1))m$.
  Since $\bigvee_{j \in [m]} C_j$ and $\bigvee_{j \in [m]} \overline{C_j}$ have the same lower bound on the number of unsatisfied clauses, $\mathcal V$ leaves at least $(1/24-o(1))m$ clauses $\overline{C_j}$ unsatisfied.
  Therefore if $|\mathcal F| \geqslant (1-1/48+o(1))m=(47/48+o(1))m$ (case 2a), then at least $(1/48-o(1))m$ clauses $\overline{C_j} \lor z_j$ (thus $C_j \lor \neg z_j$) are not satisfied.
  This implies that $\mathcal V'$ leaves $2(1/48-o(1))m=(1/24-o(1))m$ clauses of $\rho(\phi)$ unsatisfied.
  Otherwise (case 2b), it holds that $(1/48-o(1))m \leqslant |\mathcal F| \leqslant (47/48-o(1))m$.
  Thus $F$ is a non-trivial proper subset of $[m]$ whose size and complement-size are at least $(1/48-o(1))m$.
  By \cref{thm:expander} this implies that there are at least $0.46 \cdot (1/48-o(1))m$ clauses $z_a \lor \neg z_b$ which are not satisfied (those with $a \in \mathcal T$ and $b \in \mathcal F$, or with $a \in \mathcal F$ and $b \in \mathcal T$).
  In all three cases (1, 2a, 2b), at least $9m/1000 = 9m'/5000$ clauses of $\rho(\phi)$ are not satisfied.
  In other words, at most $4991 m'/5000$ clauses of $\rho(\phi)$ are satisfiable.
   
  We assume that there is an algorithm $\mathcal A$ that distinguishes in time $2^{{n'}^{1-\delta}}$ satisfiable \fsnaeb-instances from instances where at most $4991 m'/5000$ clauses can be satisfied.
  We restrict the inputs $\psi$ of \tsatb to be of the two kinds described in \cref{thm:hardness-tsat-B} (or in the two last paragraphs), and we run $\mathcal A$ on $\rho(\phi)$.
  The two previous paragraphs prove (in this order) that if $\mathcal A$ detects that at most $4991 m'/5000$ clauses can be satisfied, then at most $(7/8+o(1))M$ clauses of the \tsatb-instance are satisfiable, and if $\mathcal A$ detects that the instance is satisfiable, then the \tsatb-instance is also satisfiable.
  Finally the running time of $\mathcal A$ in terms of $N$ is $2^{{((B+1)(1+2B)N)}^{1-\delta}}=O(2^{N^{1-\frac{\delta}{2}}})$.
  Hence such an algorithm $\mathcal A$ would refute the ETH, by \cref{thm:hardness-tsat-B}.

  This completes the proof.
  We observe that if $\mathcal A$ reports a satisfying assignment $\mathcal V$ for $\rho(\phi)$, one can easily obtain a satisfying assignment for $\psi$.
  All the $4m$ clauses $z_a \lor \neg z_b$ being all satisfied, it holds that $z_1, z_2, \ldots, z_m$ have the same truth value.
  Since flipping the truth value of each variable in a satisfying Not-All-Equal-assignment results in another satisfying assignment, we can further assume that $z_1, z_2, \ldots, z_m$ are all set to true by $\mathcal V$.
  The clause $C_j \lor \neg z_j$ being satisfied, $\mathcal V$ sets to true at least one literal of $C_j$.
  Hence $\mathcal V$ restricted to the original variables (all the variables but the $z_j$ and the $d_i$) satisfies $\psi$.
\end{proof}

We now decrease the size of the clauses to at most 3.
The next reduction and the subsequent one are folklore.
We give complete proofs both for the sake of self-containment and to report explicit inapproximability bounds. 

\begin{theorem}\label{thm:hardness-nae-B}
   Under the ETH, for every $\delta > 0$ one cannot distinguish in time $2^{n^{1-\delta}}$, $n$-variable $m$-clause \nae-instances that are satisfiable from instances where at most $9991 m/10000$ clauses can be satisfied, even when each variable appears in at most $B$ clauses.
   Thus \naeb cannot be $9991/10000$-approximated in time $2^{n^{1-\delta}}$.
\end{theorem}

\begin{proof}
  We give a linear reduction $\rho$ from \fnaeb to \naeb.
  Let $\phi$ be any instance of \fnaeb with $N$ variables and $M$ clauses.
  By duplicating an arbitrary literal in the clauses on initially less than four literals, we can assume that every clause of $\phi = \bigvee_{j \in [M]} C_j$ is a 4-clause.  
  We replace each 4-clause $C_j := \ell_{j,1} \lor \ell_{j,2} \lor \ell_{j,3} \lor \ell_{j,4}$ of $\phi$ by two 3-clauses: $D_j := \ell_{1,j} \lor \ell_{2,j} \lor z_j$ and $D'_j := \ell_{3,j} \lor \ell_{4,j} \lor \neg z_j$ where $z_j$ is a fresh variable.
  This finishes the construction of the instance $\rho(\phi)$ on $n := N+M \leqslant (1+B)N$ variables and $m := 2M \leqslant 2BN$ clauses.
  We observe that $\rho(\phi)$ has only 3-clauses, and maximum occurrence of variables bounded by $B$.
  Hence $\rho(\phi)$ is an instance of \naeb.

  If $\phi$ is satisfiable then there is an assignment $\mathcal V$ of its variables such that for every $j \in [m]$, the literals $\ell_{j,1}, \ell_{j,2}, \ell_{j,3}, \ell_{j,4}$ are not all equal.
  We can then set every variable $z_j$ in the following way to satisfy every $D_j$ and $D'_j$.
  Either $\ell_{j,1}$ and $\ell_{j,2}$ are different (case 1), and thus $D_j$ is satisfied.
  In that case we can set $z_j$ so that $\neg z_j$ is opposite to $\ell_3$ to also satisfy $D'_j$.
  Either $\ell_{j,3}$ and $\ell_{j,4}$ are different (case 2), and thus $D'_j$ is satisfied.
  In that case we can set $z_j$ to the opposite sign of $\ell_1$ to also satisfy $D_j$.
  Or finally $\ell_{j,1}$ and $\ell_{j,2}$ are equal, and $\ell_{j,3}$ and $\ell_{j,4}$ are equal (case 3).
  Since $\mathcal V$ satisfies $C_j$, it cannot be that they are all equal.
  Hence $\ell_{j,1} = \ell_{j,2} = \neg \ell_{j,3} = \neg \ell_{j,4}$.
  Thus setting $z_j$ to the opposite sign of $\ell_{j,1}$ satisfies both $D_j$ and $D'_j$.
  These three cases span all the possibilities, and in each alternative there is a value for $z_j$ to satisfy the clauses $D_j$ and $D'_j$.

  We now assume that at most $4991 M/5000$ clauses of $\phi$ can be satisfied.
  For any assignment $\mathcal V'$ of the variables of $\rho(\phi)$, let $\mathcal V$ be its restriction to the variables of $\phi$.
  There are at least $9M/5000$ indices $j$ for which $\mathcal V$ sets $\ell_{j,1}, \ell_{j,2}, \ell_{j,3}, \ell_{j,4}$ to the same value.
  For such an index $j$, $\mathcal V'$ sets either $\ell_{j,1}, \ell_{j,2}, z_j$ or $\ell_{j,3}, \ell_{j,4}, \neg z_j$ to the same value.
  Hence at least $9M/5000=9m/10000$ clauses of $\rho(\phi)$ are not satisfied.
  Equivalently, at most $9991m/10000$ clauses of $\rho(\phi)$ can be satisfied.
  
  We assume that there is an algorithm $\mathcal A$ that distinguishes in time $2^{n^{1-\delta}}$ satisfiable \snaeb-instances from instances where at most $9991 m/10000$ clauses can be satisfied.
  We restrict the inputs $\phi$ of \fsnaeb to be of the two kinds described in \cref{thm:hardness-fnae-B} (or in the two last paragraphs), and we run $\mathcal A$ on $\rho(\phi)$.
  The two previous paragraphs prove that if $\mathcal A$ detects that at most $9991 m/10000$ clauses can be satisfied, then at most $4991 M/5000$ clauses of the \fsnaeb-instance are satisfiable, and if $\mathcal A$ detects that the instance is satisfiable, then the \fsnaeb-instance is also satisfiable.
  Finally the running time of $\mathcal A$ in terms of $N$ is $2^{{((1+B)N)}^{1-\delta}}=O(2^{N^{1-\frac{\delta}{2}}})$.
  Hence such an algorithm $\mathcal A$ would refute the ETH, by \cref{thm:hardness-fnae-B}.
\end{proof}

Finally, by a linear reduction from \naeb to \pnaet, we decrease the maximum number of occurrences per variable to 3, and we remove the negative literals.
A compact yet weaker implication of the following theorem is that a QPTAS for \pnaet would disprove the ETH. 

\begin{theorem}\label{thm:hardness-nae-3}
  Under the ETH, for every $\delta > 0$ one cannot distinguish in time $2^{n^{1-\delta}}$, $n$-variable $m$-clause \pnaet-instances that are satisfiable from instances where at most $\gamma m$ clauses can be satisfied, with $\gamma := (60000 B^2 - 9)/(60000 B^2)$.
   Thus \pnaet cannot be $\gamma$-approximated in time $2^{n^{1-\delta}}$.
\end{theorem}

\begin{proof}
  We give a linear reduction $\rho$ from \naeb to \pnaet.
  Let $N$ and $M \leqslant BN$ be the number of variables and clauses of the \snaeb-instance $\phi$.
  For every variable $x_i$ of $\phi$, we introduce $2B$ variables $x_{i,1}, \ldots, x_{i,B}, y_{i,1}, \ldots, y_{i,B}$ in $\rho(\phi)$, the $B$ 2-clauses $x_{i,h} \lor y_{i,h}$ for every $h \in [B]$, and the $B-1$ 2-clauses $y_{i,h} \lor x_{i,h+1}$ for every $h \in [B-1]$.
  We call \emph{equality clauses} these $(2B-1)N$ clauses, since satisfying all of them forces every $x_{i,1}, \ldots, x_{i,B}$ to be all equal, and every $y_{i,1}, \ldots, y_{i,B}$ to have the opposite value. 
  Finally we replace each clause $C_j$, say, $x_a \lor \neg x_b \lor x_c$ of the \snaeb-instance by $x_{a,\alpha} \lor y_{b,\beta} \lor x_{c,\gamma}$ where each variable $x_{i,h}$ and $y_{i,h}$ is used only once, outside the equality clauses.
  This is possible since the maximum number of occurrences of $x_i$ in $\phi$ is $B$.
  We call \emph{original clauses} these relabeled clauses.
  This finishes the reduction.
  The produced \snaet-instance $\rho(\phi)$ has $n := 2BN$ variables and $m := M+(2B-1)N \leqslant 6BM$ clauses (the last inequality holds since $N \leqslant 3M$).
  There is no negation in $\rho(\phi)$ and each variable appears at most three times: in two equality clauses and one original clause.

  If $\phi$ is satisfiable, then $\rho(\phi)$ is also satisfiable by setting the $x_{i,h}$ to the same value as $x_i$, and the $y_{i,h}$ to the opposite value.
  That way all the equality clauses are satisfied (they have one true and one false literal).
  The original clauses are satisfied the same way $\phi$ was satisfied.

  We now assume that at most $9991 M/10000$ clauses of $\phi$ can be satisfied.
  Let $\mathcal V'$ be any assignment of the variables of $\rho(\phi)$.
  We do the following thought experiment.
  We start with an assignment $\mathcal V$ satisfying all the equality clauses, and respecting for $x_{i,1}, \ldots, x_{i,B}$ the majority choice of $\mathcal V'$ on these variables (breaking ties arbitrarily).
  By assumption $\mathcal V$ does not satisfy at least $9M/10000$ original clauses.
  Then we switch one by one the value of the variables disagreeing with $\mathcal V'$ (until we reach $\mathcal V'$).
  At the cost of one unsatisfied equality clause, we can fix $B$ original clauses.
  Eventually $\mathcal V'$ leaves at least $9M/(10000 B) \geqslant 9m/(60000 B^2)$ clauses unsatisfied.
  Thus any assignment of $\rho(\phi)$ satisfies at most $\frac{60000 B^2 - 9}{60000 B^2}m=\gamma m$.
  We recall that $B$ is a function of the value $s_3$ supposed positive by the ETH.

  We assume that there is an algorithm $\mathcal A$ that distinguishes in time $2^{n^{1-\delta}}$ satisfiable \pnaet-instances from instances where at most $\gamma m$ clauses can be satisfied.
  We restrict the inputs $\phi$ of \naeb to be of the two kinds described in \cref{thm:hardness-nae-B} (or in the two last paragraphs), and we run $\mathcal A$ on $\rho(\phi)$.
  The two previous paragraphs prove that if $\mathcal A$ detects that at most $\gamma m$ clauses can be satisfied, then at most $9991 M/10000$ clauses of the \snaeb-instance are satisfiable, and if $\mathcal A$ detects that the instance is satisfiable, then the \snaeb-instance is also satisfiable.
  Finally the running time of $\mathcal A$ in terms of $N$ is $2^{{(2BN)}^{1-\delta}}=O(2^{N^{1-\frac{\delta}{2}}})$.
  Hence such an algorithm $\mathcal A$ would refute the ETH, by \cref{thm:hardness-nae-B}.
\end{proof}

This last reduction no longer implies APX-hardness.
Indeed, the value $B$ in the inapproximability ratio is finite only if $s_3 > 0$.
So one should assume the ETH, and not the mere P $\neq$ NP, to rule out an approximation algorithm with ratio $\gamma < 1$.
Sacrificing the strong lower bound in the running time, we can overcome that issue.
Berman and Karpinski \cite{Berman01} showed that it is NP-hard to approximate \textsc{Max 2-SAT-3} within ratio better than $787/788$.
Following the reduction of \cref{thm:hardness-fnae-B} from \textsc{Max 2-SAT-3}, we derive the following inapproximability.

\begin{corollary}\label{cor:hardness-nae-B}
  Approximating \textsc{NAE 3-SAT-9} within ratio $51326/51327$ is NP-hard.
\end{corollary}
\begin{proof}
  Observe that the clause size grows from 2 to 3, and that the variables $z_j$ are part of at most 9 clauses.
\end{proof}

Then following \cref{thm:hardness-nae-B}, we get: 
\begin{corollary}\label{cor:hardness-nae-3}
  Approximating \psnaet within ratio $49888956/49888957$ is NP-hard.
\end{corollary}

\section{\mipa, unit disks and rectangles}\label{sec:disks-and-rectangles}

We introduce the \mipa-problem (\smipa, for short), a convenient intermediate problem to show APX-hardness for geometric problems.
We start with an informal description.
Let $M$ be a perfect matching between the $n$ points $[n] \times \{0\}$ and $[n] \times \{1\}$, in $\mathbb N^2$.
This matching can be represented by a permutation $\sigma$, such that for every $i \in [n]$, $(i,0)$ is matched with $(\sigma(i),1)$.
Imagine now a set of intervals on the line $y = 1/2$ whose endpoints are all in $[n]$.
The aim is to move each interval "up" or "down", by translating it by $(0,1/2)$ or by $(0,-1/2)$, respectively, such that the number of edges of $M$ with no endpoint on a translated interval is maximized.
An edge of $M$ with at least one endpoint in a \emph{moved} (or \emph{positioned}) interval is said \emph{covered} or \emph{destroyed}.
The edge is said \emph{uncovered} or \emph{preserved} otherwise.
Equivalently \mipa aims to minimizing the number of covered edges, or maximizing the number of uncovered edges.
We choose the maximization formulation, since we will both reduce from a maximization problem (\pnaet) and to a maximization problem (\cli on disks and rectangles).
Thus the objective value will be the number of \emph{uncovered edges}.

\defproblem{\mipa}{A permutation $\sigma$ over $[n]$ representing a perfect matching $M$ between the points $(1,0), (2,0), \ldots, (n,0)$ and $(\sigma(1),1), (\sigma(2),1), \ldots, (\sigma(n),1)$ respectively, and a set of integer ranges $\mathcal I := \{I_1, \ldots, I_h\}$ where $I_k := [\ell_k,r_k]$ and $1 \leqslant \ell_k \leqslant r_k \leqslant n$, for every $k \in [h]$.}{A placement function $p: \mathcal I \rightarrow \{0,1\}$ encoding that interval $I_k$ has its endpoints positioned in $(\ell_k,p(I_k))$ and $(r_k,p(I_k))$, which maximizes the number of edges of $M$ with no endpoint on a positioned interval.}

\begin{figure}[h!]
  \centering
  \begin{tikzpicture}
    \def\n{10}
    \draw (0,0) -- (\n+1,0) ;
    \draw (0,1) -- (\n+1,1) ;
    \foreach \i in {1,...,\n}{
      \node[vp] (p\i) at (\i,0) {} ;
      \node[vp] (q\i) at (\i,1) {} ;
      \node at (\i,-0.3) {$\i$} ;     
    }
    \foreach \i/\j in {3/2,2/3,1/5,5/1,4/8,8/4,6/10,10/6,7/9,9/7}{
      \draw[blue] (p\i) -- (q\j) ;
    }
    \foreach \i/\j/\c in {1/3/red,4/7/orange,8/10/black!15!yellow}{
      \draw[very thick, \c] (\i,0.5) -- (\j,0.5) ;
      \draw[very thick, \c] (\i,0.6) -- (\i,0.4) ;
      \draw[very thick, \c] (\j,0.6) -- (\j,0.4) ;
    }
  \end{tikzpicture}
  \caption{An example of a symmetric instance of \smipa with three disjoint ranges.}
  \label{fig:example-mipa}
\end{figure}
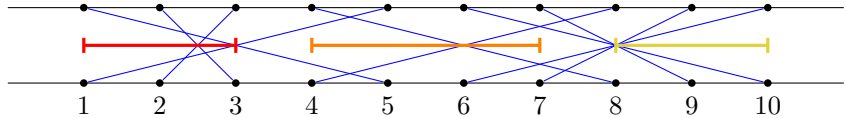

A \smipa-input may interchangeably be given as $(\sigma,\mathcal I)$ or as $(M,\mathcal I)$.
One may observe that a \emph{constant} placement (i.e., $p(I_1)=\ldots=p(I_h)=0$, or $p(I_1)=\ldots=p(I_h)=1$) is a worse solution when the intervals of $\mathcal I$ span $[n]$, since it covers all the edges of $M$.
We say that the matching $M$ is \emph{symmetric} if $(i,0)(j,1) \in M$ implies that $(i,1)(j,0) \in M$, for every $i, j \in [n]$; in the geometric viewpoint, it is equivalent to $y=1/2$ being a symmetry axis of $M$, and in the permutation viewpoint, it is equivalent to $\sigma$ being a product of pairwise-disjoint transpositions.
Other handy (as far as hardness of geometric problems is concerned) technical problems involving intervals and/or permutations include \textsc{Crossing-Avoiding Matching} in Gu\'{s}piel~\cite{Guspiel17} or \textsc{Crossing-Minimizing Perfect Matching} in Gu\'{s}piel et al.~\cite{Guspiel19}, the problem of covering a 2-track point set by selecting $k$ 2-track intervals \cite{MarxP15} or \textsc{Structured 2-Track Hitting Set}~\cite{BonnetM16}.
It is no coincidence that these convenient starting problems all involve matchings/permutations and/or intervals.
Indeed the latter objects are more easily encoded in a geometric setting than their generalizations: arbitrary binary relations and arbitrary sets.
Later we will see how disks can encode intervals and how rectangles can encode a permutation, in the context of the \cli-problem.

We rule out an approximation scheme for \mipa, even if subexponential-time is allowed.
In particular a QPTAS for \smipa is highly unlikely.
We recall that $\gamma = (60000 B^2 - 9)/(60000 B^2)$ and that $B$ is a finite integral constant, assuming the ETH ($s_3>0$).

\begin{lemma}\label{lem:mipa}
  For every $\delta > 0$, \mipa cannot be $\gamma'$-approximated in time $2^{|M|^{1-\delta}}$, with $\gamma' := 1 - (1-\gamma)/13 < 1$, unless the ETH fails.
  Besides \mipa is NP-hard and APX-hard.
  These results hold even if the length of every interval of $\mathcal I$ is at most 5, and the matching $M$ is symmetric.
\end{lemma}

\begin{proof}
  We give a reduction $\phi$ from \pnaet to \mipa.
  Let $\phi$ be a \psnaet-instance, with variables $x_1, \ldots, x_n$ and clause $C_1, \ldots, C_m$.
  For every $x_i \in C_j$, we denote by $\text{occ}(x_i,C_j)$ the number of occurrences of $x_i$ in the clauses $C_1, \ldots, C_j$.
  We observe that $\text{occ}(x_i,C_j) \in \{1,2,3\}$.
  We build an instance $\rho(\phi) := (M, \mathcal I)$ of \smipa in the following way.
  For each variable $x_i$ of $\phi$, we reserve a range $[3(i-1)+1,3(i-1)+3]$ with 3 integral points on both lines $y=0$ and $y=1$.
  These points will be matched by $M$ to points in the clause gadgets.
  We add the interval $X_i := [3(i-1)+1,3(i-1)+3]$ to $\mathcal I$.
  We now describe the 2-clause and the 3-clause gadgets.
  
  For every 2-clause $C_j := x_a \lor x_b$, we allocate a slot $S_j$ of size 9 (on $y=0$ and $y=1$) appended to the current last position.
  The first half of $S_j$, that is, the indices in $[s_j, s_j+4]$ of $S_j$ correspond to $x_a$, while the indices in $[s_j+5, s_j+9]$ correspond to $x_b$.
  For every $(d_1,d_2) \in \{(0,1), (1,0)\}$ and $h \in [4]$, we add to $M$ the edge between $(s_j+h,d_1)$ and $(s_j+5+h,d_2)$.
  We add the intervals $C_j(x_a) := [s_j, s_j+4]$ and $C_j(x_b) := [s_j+5, s_j+9]$ to $\mathcal I$.
  Finally for each $(d_1,d_2) \in \{(0,1), (1,0)\}$, we add to $M$ the edges between $(s_j,d_1)$ and $(3(a-1)+\text{occ}(x_a,C_j),d_2)$, and between $(s_j+3,d_1)$ and $(3(b-1)+\text{occ}(x_b,C_j),d_2)$.

  For every 3-clause $C_j := x_a \lor x_b \lor x_c$, we allocate a slot $S_j$ of size 15 (on $y=0$ and $y=1$) appended to the current last position.
  The first third of $S_j$, that is, the indices in $[s_j, s_j+4]$ of $S_j$ correspond to $x_a$, the second third, the indices in $[s_j+5, s_j+9]$ correspond to $x_b$, and the last third, the indices in $[s_j+10,s_j+14]$ correspond to $x_c$.
  We add the intervals $C_j(x_a) := [s_j, s_j+4]$, $C_j(x_b) := [s_j+5, s_j+9]$, and $C_j(x_c) := [s_j+10, s_j+14]$ to $\mathcal I$.
  Similarly for every $(d_1,d_2) \in \{(0,1), (1,0)\}$ and $(h,p) \in \{(a,0),(b,1),(c,2)\}$, we add to $M$ the edge between $(s_j+5p,d_1)$ and $(3(h-1)+\text{occ}(x_h,C_j),d_2)$.
  We call these edges \emph{internal} (same for the 2-clause gadget).
  Finally we add to $M$ four edges from every pair of ranges in $\{[s_j, s_j+4],[s_j+5, s_j+9],[s_j+10, s_j+14]\}$, two starting on the line $y=0$ (ending on $y=1$) and two starting on $y=1$ (ending on $y=0$).
  We call these edges \emph{variable-clause} (same for the 2-clause gadget).

  For each variable $x_i$ with only two occurrences in $\phi$, we link its third occurrence pair to a dummy pair $(d_i,0), (d_i,1)$, appended to the current last position.
  That is, we add the edges $(3(i-1)+3,0)(d_i,1)$ and $(3(i-1)+3,1)(d_i,0)$ to $M$.
  Although not needed, we also add the singleton interval $D_i := \{d_i\}$ to $\mathcal I$.
  We call it \emph{dummy gadget} and consider it as a special case of a clause gadget.
  This finishes the construction of the \smipa-instance $(M,\mathcal I)$.
  Observe that every point is matched, and that all the intervals of $\mathcal I$ are pairwise disjoint, and of length at most 5.
  The perfect matching $M$ comprises at most $3n+15m+n \leqslant 49n$ edges.

  We assume that $\phi$ is satisfiable, and let $\mathcal V$ be a satisfying assignment.
  We build the following solution to the \smipa-instance.
  We push the interval $X_i$ up (to the line $y=1$) if $x_i$ is set to true by $\mathcal V$, and we push it down (to the line $y=0$) otherwise.
  In the clause gadgets (and dummy gadgets), we do the opposite: we push $C_j(x_i)$ ($D_i$) down if $x_i$ is set to true, and up if $x_i$ is set to false.
  This solution preserves four edges within each clause gadget, and an additional $3n$ edges between the variable gadgets and the clause gadgets.
  Hence the total number of preserved edges is $4m+3n$.

  We now assume that at most $\gamma m$ clauses of $\phi$ are satisfiable.
  Let $p$ be a placement function of the intervals of $\mathcal I$, maximizing the number of preserved edges of $M$.
  We first argue that not giving the same placement (up/1 or down/0) to the three (resp. two) intervals $C_j(x_a), C_j(x_b), C_j(x_c)$ (resp. $C_j(x_a), C_j(x_b)$) of a 3-clause gadget (resp. 2-clause gadget) is always better.
  Note that any equal placement destroys all the edges of $M$ internal to the clause gadget of $C_j$, and preserves at most three variable-clause edges.
  On the other hand, a placement with at least one interval on each side preserves already four internal edges.
  We can then assume that $p$ does not give equal placement in any clause gadget.
  Let $\mathcal V$ be the assignment of the variables of $\phi$ which sets $x_i$ to true if $p(X_i)=1$, and to false, if $p(X_i)=0$.
  By assumption $\mathcal V$ does not satisfy at least $(1-\gamma)m$ clauses.
  In each corresponding clause gadget, one can preserve at most two variable-clause edges of $M$.
  Indeed all three variable-clause edges incident to the clause gadget and not covered by the placement of the $X_i$ land on the same side.
  By the previous remark, at least one such edge should be destroyed (to preserve four internal edges).
  Thus the placement $p$ preserves at most $3n+4m-(1-\gamma)m$ edges.

  Since $|M|=O(n+m)=O(n)$ and $\frac{3n+4m-(1-\gamma)m}{3n+4m} \leqslant 1 - \frac{1-\gamma}{13}$, by \cref{thm:hardness-nae-3} \smipa cannot be $\gamma'$-approximated in time $2^{|M|^{1-\delta}}$, under the ETH.
  Besides, by \cref{cor:hardness-nae-3}, \smipa cannot be $648556435/648556436$-approximated in polynomial-time, unless P$=$NP.
  In particular, this problem is NP-hard and even APX-hard.
\end{proof}

We recall that \cli can be solved in polynomial-time in unit disk graphs~\cite{Clark90,Raghavan03} and in axis-parallel rectangle intersection graphs~\cite{Brimkov18}.
Now if the objects can be unit disks \emph{and} axis-parallel rectangles, we show that even a SUBEXPAS is unlikely.
We denote by \textsc{\{Obj,Obj'\}}-\cli the clique problem in the intersection graphs of objects that can be either \textsc{Obj} or \textsc{Obj'}. 

\begin{theorem}\label{thm:cli}
  For every $\delta > 0$, \cli in intersection graphs $G$ of unit disks and axis-parallel rectangles cannot be $c$-approximated in time $2^{{|V(G)|}^{1-\delta}}$, with $c := 1 - (1-\gamma)/153 <~1$, unless the ETH fails.
  Besides this problem is NP-hard and APX-hard.
\end{theorem}
\begin{proof}
  We give a reduction from \mipa to \textsc{\{Unit Disks, Axis-Parallel Rectangles\}}-\cli or \textsc{\{Half-Planes, Axis-Parallel Rectangles\}}-\cli.
  Let $(M,\mathcal I)$ be an instance of \smipa over $[n]$, where $M$ is symmetric, and all the intervals of $\mathcal I$ have size at most 5.
  We build the following set of axis-parallel rectangles $\mathcal R$ and half-planes $\mathcal H$.
  See Figure~\ref{fig:mipa-to-clique} for an illustration.

  Let $O$ be the origin of the plane.
  We place from left to right $n+2$ points $p_0, p_1, \ldots, p_n, p_{n+1}$ on a convex curve in the top-left quadrant, say $x \mapsto -1/x$ on $[-(1+\lambda),-1]$ for some small $\lambda > 0$.
  We wiggle the points $p_i$ so that for every $i \leqslant j \in [n]$, the slope of the line passing through $\midd(p_{i-1},p_i)$ and $\midd(p_j,p_{j+1})$ has a distinct value.
  We define $q_0, q_1, \ldots, q_n, q_{n+1}$, such that $O$ is the middle of the segment $p_iq_i$ for every $i \in [0,n+1]$.
  In other words, this new chain is obtained by central symmetry about $O$.
  Observe that sorted by $x$-coordinates, these $2n+4$ points read $p_0, p_1, \ldots, p_n, p_{n+1}, q_{n+1}, q_n, \ldots, q_1, q_0$.
  The points $p_1, \ldots, p_n$ represent $[n] \times \{0\}$ in the \smipa-instance, while the points $q_1, \ldots, q_n$ represent $[n] \times \{1\}$.

  For every pair $i \leqslant j \in [n]$, we can associate a line $\ell_p(i,j)$ passing through $\midd(p_{i-1},p_i)$ and $\midd(p_j,p_{j+1})$.
  Notice that, by convexity, $\ell_p(i,j)$ separates the points $p_i, p_{i+1}, \ldots, p_{j-1},$ $p_j$ (below it) from the points $p_1, \ldots, p_{i-1},$ $p_{j+1}, \ldots, p_n$ (above it).
  We similarly define $\ell_q(i,j)$ as the line passing through $\midd(q_{i-1},q_i)$ and $\midd(q_j,q_{j+1})$.
  We observe that $\ell_p(i,j)$ and $\ell_q(i,j)$ are parallel.
  For every interval $I = [i,j] \in \mathcal I$, we introduce in the \cli-instance the half-plane $h_p(I) := h_p(i,j)$ as the closed upper half-plane whose boundary is $\ell_p(i,j)$, and $h_q(I) := h_q(i,j)$ as the closed lower half-plane whose boundary is $\ell_q(i,j)$.
  We give these two objects weight 5 by superimposing 5 copies of them.
  All pairs of introduced half-planes intersect, except the pairs $\{h_p(i,j), h_q(i,j)\}$.

  Finally for every edge $(i,0)(j,1)$ of the matching $M$ (with $i, j \in [n]$), we add an axis-parallel rectangle $R(i,j)$ whose top-left corner is $p_i$ and bottom-right corner is $q_j$.
  This finishes the construction of $(\mathcal R, \mathcal H)$.
  When $\lambda$ tends to $0$, the rectangles are arbitrary close to squares of equal side-length.
  In other words, for any $\varepsilon > 0$, the axis-parallel rectangles can be made $\epsilon$-squares.
  The half-planes can be turned into unit disks, making the side-length of the rectangles very small compared to 1.
  We denote by $(\mathcal R, \mathcal D)$ the corresponding sets of axis-parallel rectangles and unit disks, and by $G$ their intersection graph.

  \begin{figure}[h!]
    \centering
    \begin{tikzpicture}
      \def\n{10}
      \def\s{4}
      \def\sy{8}
      \foreach \i/\j in {3/2,2/3,1/4,4/1,5/8,8/5,6/10,10/6,7/9,9/7}{
        \pgfmathsetmacro\pxi{- \s * \i / \n}
        \pgfmathsetmacro\pyi{\sy / (\i+1)}
        \pgfmathsetmacro\qxj{\s * \j / \n}
        \pgfmathsetmacro\qyj{- \sy / (\j+1)}
        \fill[opacity=0.2,blue] (\pxi,\pyi) -- (\pxi,\qyj) -- (\qxj,\qyj) -- (\qxj,\pyi) -- cycle ;
      }
      \def\sq{5}
      \def\xmin{-\sq}
      \def\xmax{\sq}
      \def\ymin{-\sq}
      \def\ymax{\sq}

      \foreach \si in {1,-1}{
        \fill[opacity=0.2,red] (\xmax * \si,1.75 * \si) -- (\xmin * \si,0.75 * \si) -- (\xmin * \si,\ymax * \si) -- (\xmax * \si,\ymax * \si) -- cycle ;
        \fill[opacity=0.2,orange] (\xmax * \si,5 * \si) -- (\xmin * \si,-0.08 * \si) -- (\xmin * \si,\ymax * \si) -- (\xmax * \si,\ymax * \si) -- cycle ;
        \fill[opacity=0.2,yellow] (-0.1 * \si,\ymax * \si) -- (-4 * \si,\ymin * \si) -- (\xmin * \si,\ymin * \si) -- (\xmin * \si,\ymax * \si) -- cycle ;
      }
      \foreach \si in {1,-1}{
        \draw[very thick, yellow] (-0.1 * \si,\ymax * \si) -- (-4 * \si,\ymin * \si) ;
        \draw[very thick, orange] (\xmax * \si,5 * \si) -- (\xmin * \si,-0.08 * \si) ;
        \draw[very thick,red] (\xmax * \si,1.75 * \si) -- (\xmin * \si,0.75 * \si) ;
      }
      \foreach \i in {1,...,\n}{
        \pgfmathsetmacro\px{- \s * \i / \n}
        \pgfmathsetmacro\py{\sy / (\i + 1)}
        \pgfmathsetmacro\qx{-\px}
        \pgfmathsetmacro\qy{-\py}
        \node[vp] (p\i) at (\px,\py) {} ;
        \pgfmathtruncatemacro\ii{11-\i}
        \node (tp\i) at (\px,\py + 0.2) {\footnotesize{$p_\ii$}} ;
        \node[vp] (q\i) at (\qx,\qy) {} ;
        \node (tq\i) at (\qx,\qy - 0.2) {\footnotesize{$q_\ii$}} ;
      }
    \end{tikzpicture}
    \caption{The output of the reduction on the instance of \cref{fig:example-mipa}.}
    \label{fig:mipa-to-clique}
  \end{figure}
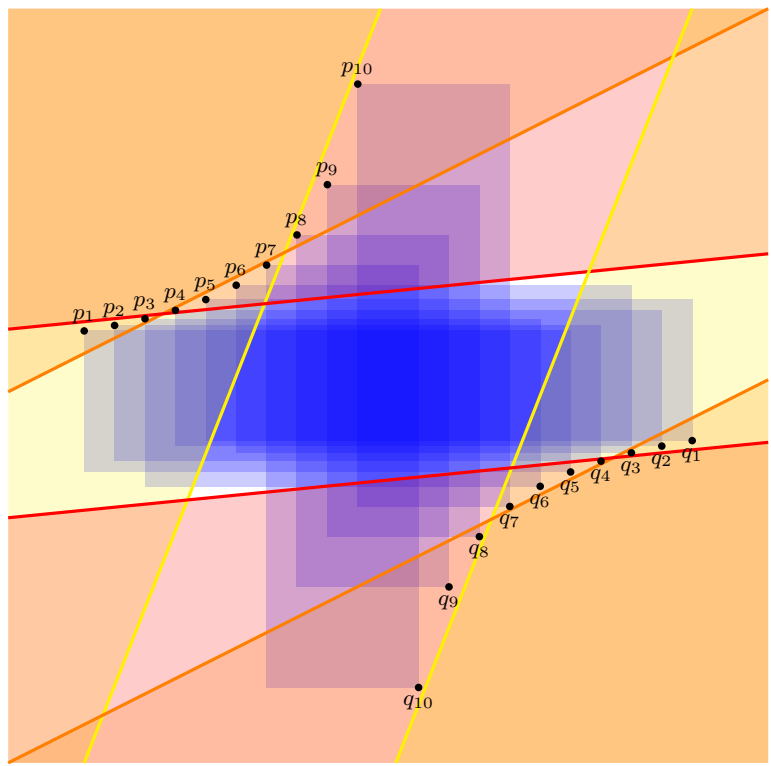
  
  Let consider instances of \smipa produced by the previous reduction from \psnaet, on $\nu$-variable $\mu$-clause formulas that are either satisfiable or with at least $(1-\gamma)\mu$ non satisfiable clauses.
  We call \emph{yes-instances} the former \smipa-instances, and \emph{no-instances}, the latter.
  If $(M,\mathcal I)$ is a yes-instance, we claim that $G$ has a clique of size $5|\mathcal I|+3 \nu+4 \mu$.
  Indeed there is a placement $p$ that preserves $3 \nu+4 \mu$ edges of $M$.
  We start by taking in the clique all the half-planes (or corresponding unit disks) $h_p(I)$ whenever $p(I)=0$, and $h_q(I)$ whenever $p(I)=1$.
  Since these objects have weight 5 (actually 5 stacked copies), this amounts to $5|\mathcal M|$ vertices.
  The corresponding half-planes pairwise intersect since their boundaries have distinct slopes.
  Then we include to the clique the $3 \nu+4 \mu$ rectangles $R(i,j) \in \mathcal R$ such that $(i,0)(j,1)$ is preserved by $p$.
  All the rectangles pairwise intersect since they all contain the origin $O$.
  Every pair of chosen half-plane $h_z(I)$ ($z \in \{p,q\}$) and rectangle $R(a,b)$ intersects, otherwise the placement of $I$ would cover $(a,0)(b,1)$.
  Thus we exhibited a clique of size $5|\mathcal I|+3 \nu+4 \mu$ in $G$.

  We now assume that $(M,\mathcal I)$ is a no-instance, and we claim that $G$ has no clique larger than $5|\mathcal I|+3 \nu+4 \mu - (1-\gamma)\mu$.
  Let us see how to build a clique in $G$.
  One can take at most one object between $h_p(I)$ and $h_q(I)$ (since they do not intersect).
  There is a maximum clique that takes at least one of $h_p(I)$ and $h_q(I)$ since $h_p(I)$ has weight 5 and intersects every object but $h_q(I)$ plus at most 5 rectangles (recall that the intervals of $\mathcal I$ have size at most 5).
  Thus we assume that our maximum clique takes exactly one object between $h_p(I)$ and $h_q(I)$, for every $I \in \mathcal I$.
  We consider the placement $p$ defined as $p(I)=0$ if $h_p(I)$ is in the clique, and $p(I)=1$ if $h_q(I)$ is in the clique.
  Now the rectangles $R(i,j)$ that are adjacent to the chosen half-planes of $\mathcal H$ (or unit disks of $\mathcal D$) correspond to the edges $(i,0)(j,1)$ of $M$ which are preserved.
  By \cref{lem:mipa}, there are at most $3 \nu+4 \mu - (1-\gamma)\mu$ such rectangles.

  Since $|V(G)|=|\mathcal H|+|\mathcal R| = 10|\mathcal I|+|M|=O(\nu+\mu)=O(\nu)$ and $\frac{5|I|+3 \nu+4 \mu-(1-\gamma)\mu}{|I|+3 \nu+4 \mu} \leqslant 1 - \frac{(1-\gamma)\mu}{140 \mu + 9 \mu + 4 \mu} = 1 - \frac{1-\gamma}{153}=c$, by \cref{thm:hardness-nae-3}, \textsc{\{Half-Planes/Unit Disks, Axis-Parallel Rectangles\}}-\cli cannot be $c$-approximated in time $2^{|V(G)|^{1-\delta}}$, under the ETH.
  Besides, by \cref{cor:hardness-nae-3}, this problem cannot be $7633010347/7633010348$-approximated in polynomial-time, unless P$=$NP.
  In particular, it is NP-hard and even APX-hard.
\end{proof}

Of course the ratios that are shown not achievable, even in subexponential-time, under the ETH, are very close to 1.
The current best exact exponential algorithm solving \tsat has running time $1.308^n$ \cite{Hertli14}, building upon the PPSZ algorithm \cite{Paturi05}.
Assuming getting this down to $1.14^n$ is impossible, which implies $s_3 > 0.2$, the inapproximability bound in subexponential-time of respectively $\gamma'$ for \mipa and $c$ for \textsc{\{Half-Planes/Unit Disks, Axis-Parallel Rectangles\}}-\cli are roughly $1-6 \cdot 10^{-26}$ and $1-5 \cdot 10^{-27}$, respectively.

We observe that if all the half-planes pairwise intersect (for instance because their boundaries are assumed to have distinct slopes), then there is a polynomial-time algorithm, given a geometric representation.
Let again $\mathcal H$ be the half-planes and $\mathcal R$, the axis-parallel rectangles, in the representation of the graph $G$.
Recall that the number of maximal cliques in $G[\mathcal R]$ is polynomial, and that they can be enumerated efficiently.
For each maximal clique $\mathcal R_c \subseteq \mathcal R$, we compute the maximum clique in the co-bipartite graph $G[\mathcal H \cup \mathcal R_c]$.
This is thus equivalent to computing \mis in a bipartite graph.
Due to K\H{o}nig's theorem, this can be done in polynomial-time by a matching algorithm.
We output $C$ the largest clique that we find.
$C$ is a maximum clique in $G$, since $C \cap \mathcal R$ is by definition a clique, so it is contained in a maximal clique of $G[\mathcal R]$.

Let us briefly discuss the issue the \emph{co-2-subdivision} approach encounters for \textsc{\{Half-Planes, Axis-Parallel Rectangles\}}-\cli.
Axis-parallel rectangles cannot represent a large antimatching (they already cannot represent $\overline{3K_2}$).
Hence, as in our construction, the large antimatching has to be, for the most part, realized by half-planes.
Now in the \smipa approach, the axis-parallel rectangles can avoid \emph{two} arbitrary half-planes with the freedom of their top-left and bottom-right corners.
In the \emph{co-2-subdivision} approach, they would have to avoid \emph{at least three} arbitrary half-planes, and do not have enough degrees of freedom for that. 

\section{Translates of a convex set}\label{sec:translates}

We show in this section that we can extend the algorithm of Clark et al.~\cite{Clark90} and its robust version~\cite{Raghavan03} from unit disks to any centrally symmetric, bounded, convex set.

\begin{theorem}\label{thm:polytimeTranslatedCenSymConvex}
\cli admits a robust polynomial-time algorithm in intersection graphs of translates of a fixed centrally symmetric, bounded, convex set.
\end{theorem}

Moreover, as shown by Aamand et al.~\cite{Aamand19}, for every bounded and convex set $S_1$, there exists a centrally symmetric, bounded and convex set $S_2$ such that $\mathcal{G}_{S_1} = \mathcal{G}_{S_2}$, where $\mathcal{G}_S$ denotes the intersection graphs class of translates of $S$.
Thus we obtain the immediate corollary:

\begin{corollary}\label{cor:polytimeTranslatedConvex}
\cli admits a robust polynomial-time algorithm in intersection graphs of translates of a fixed bounded and convex set.
\end{corollary}

We prove \cref{thm:polytimeTranslatedCenSymConvex} in two steps.
First we show how to compute in polynomial time a maximum clique when a representation is given.
Secondly we use the result by Raghavan and Spinrad~\cite{Raghavan03} to obtain a robust algorithm.

Let $S$ be a centrally symmetric, bounded, convex set.
We can define a corresponding norm as follow: for any $x\in \mathbb{R}^2$, let $\|x\|$ be equal to $\inf\{\lambda>0 \mid x \in \lambda S\}$.
This is well-defined since $S$ is bounded.
It is absolutely homogeneous because $S$ is centrally symmetric, and it is subadditive because $S$ is convex.
Therefore $\|.\|$ is a norm.
Let $S_1$ and $S_2$ be two translates of $S$, with respective centers $c_1$ and $c_2$.
Remark that $S_1$ and $S_2$ intersect if and only if $\|c_1-c_2\|~\leqslant 2$.
Let us assume that $d := \|c_1-c_2\|~\leqslant 2$.
We denote by $S'$ the set $S$ scaled by~$d$: $S':=dS$, and we then define: $D:=\{x \in \mathbb{R}^2 \mid \|x-c_1\|~\leqslant d, \|x-c_2\|~\leqslant d\}$.
Equivalently we have $D=(c_1+S')\cap(c_2+S')$.
If $S$ was a unit disk, $D$ would be the intersection of two disks with radius $d$, such that the boundary of one contains the center of the other.

\begin{figure}[h!]
  \centering
\includegraphics{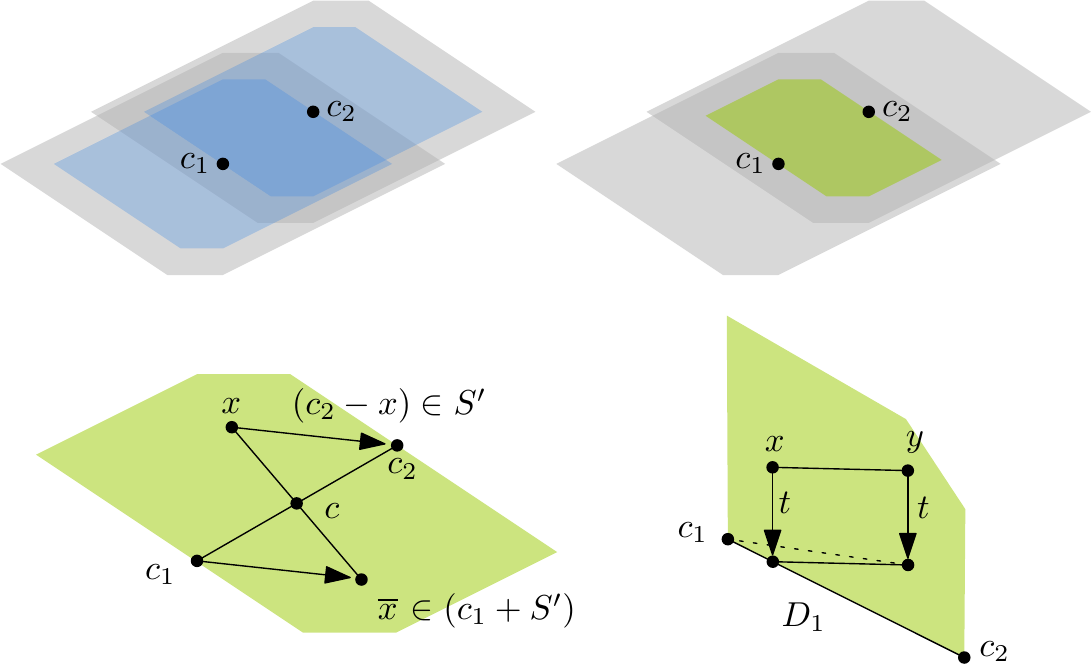}
\caption{Top left: The gray sets are scaled about their center so that the center of one set is on the boundary of the other. 
Top right: the intersection $D$. 
Bottom left: Illustration of Lemma~\ref{lem:DCentrSym}.
Bottom right: Illustration of Lemma~\ref{lem:D1clique}.}
\end{figure}

\begin{lemma}\label{lem:DCentrSym}
The set $D$ is centrally symmetric around $c:=(c_1+c_2)/2$.
\end{lemma}

\begin{proof}
Let $x$ be a point in $D$, we need to show that $\bar{x}:=x+2(c-x)$ is in $D$ too. As $D=(c_1+S')\cap(c_2+S')$, it is sufficient to show $\bar{x}\in c_1+S'$ and $\bar{x}\in c_2+S'$. By definition, $\bar{x}$ is equal to $c_1+c_2-x$. Since $x$ is in $D$, then $\|c_2-x\|~\leqslant d$, which implies that $c_2-x$ is in $S'$. Therefore $\bar{x}$ is in $c_1+S'$. By the symmetry of the arguments, we obtain that $\bar{x}$ is in $D$.
\end{proof}

\begin{lemma}\label{lem:parallelTangent}
The tangents to $D$ at $c_1$ and $c_2$ are parallel.
\end{lemma}

\begin{proof}
Let us denote by $\ell_1$ the tangent to $D$ at $c_1$. Then we denote by $\ell_2$ the line parallel to $\ell_1$ that contains $c_2$. We claim that $\ell_2$ is tangent to $D$. By construction $D$ is convex, as the intersection of two convex sets. This implies that $\ell_2$ is tangent to $D$ if and only if $D \cap \ell_2$ is a line segment that contains $c_2$. This line segment may be only one point. Let $x$ be a point in $D \cap \ell_2$. By Lemma~\ref{lem:DCentrSym}, $D$ is centrally symmetric around $c$. Therefore $x+2(c-x)$ is in $D$, and by construction it is also in $\ell_1$. Since $D \cap \ell_1$ is a line segment that contains $c_1$, thus $D \cap \ell_2$ is a line segment that contains $c_2$.
\end{proof} 

We cut $D$ along the line $\ell$ going through $c_1$ and $c_2$, and split $D$ into two sets denoted by $D_1$ and $D_2$. We define $D_1$ as the set of points below this line, and $D_2$ as the set of points not below. We have the following lemma:

\begin{lemma}\label{lem:D1clique}
Let $i$ be in $\{1,2\}$, and let $x$ and $y$ be in $D_i$. Then we have $\|x-y\|~\leqslant d$.
\end{lemma}


\begin{proof}
We do the proof for $i=1$, and the case $i=2$ can be done symmetrically.
 By Lemma~\ref{lem:parallelTangent}, the tangents $\ell_1$ and $\ell_2$ 
 of $D$ at $c_1$ and $c_2$ are parallel. Without loss of generality, 
 let us assume that they are vertical, that $c_1$ is to the left of $c_2$
 and $x$ to the left of $y$.
We denote by $\tilde{x}$  (respectively $\tilde{y}$) the vertical projection of $x$ 
   (respectively $y$) on $\ell$. 
   Without loss of generality $\|x- \tilde{x}\| \leqslant \|y- \tilde{y}\|$.
	We define $t = x- \tilde{x}$. 
	Note that $\|x-y\| = \|(x-t) - (y-t)\| =  \|\tilde{x} - (y-t)\|$.
	Furthermore, we can move $\tilde{x}$ on $\ell$ towards
	$c_1$ and this will only increase the distance to $(y-t)$.
	We get $\|\tilde{x} - (y-t)\| \leqslant \|c_1 - (y-t)\|$.
	By definition $(y-t) \in D_1\subset D$ and thus 
	$\|c_1 - (y-t)\| \leqslant d$.
	This implies 	$\|x - y\| \leqslant d$ and finishes the proof.
\end{proof}

Following the arguments of Clark et al.~\cite{Clark90}, 
one first guesses in quadratic time $S_1$ and $S_2$ in a maximum clique $C$ such that the distance between their centers $\|c_1-c_2\|$ is maximized among the pairs $S_1, S_2 \in C$.
One can then remove all the objects not centered in $D$.
By~\cref{lem:D1clique}, the intersection graph induced by the sets centered in $D$ is cobipartite.
Since computing an independent set in a bipartite graph can be done in polynomial time, then one can compute a maximum clique in $G$ in polynomial time.

Before explaining how to compute a maximum clique when no representation is given, we need to introduce a few definitions.
Let $\Lambda=e_1,e_2, \dots, e_m$ be an ordering of the $m$ edges of $G$.
Let $G_\Lambda(k)$ be the subgraph of $G$ with edge-set $\{e_k,e_{k+1}, \dots, e_m\}$.
For each $e_k=uv$, $N_{\Lambda,k}$ is defined as the set of vertices adjacent to $u$ and $v$ in $G_\Lambda(k)$.

\begin{definition}[Raghavan and Spinrad~\cite{Raghavan03}]\normalfont
An edge ordering $\Lambda=e_1, e_2,\dots,e_m$ is a \emph{cobipartite neighborhood edge elimination ordering} (CNEEO), if for each $e_k$, $N_{\Lambda,k}$ induces a cobipartite graph in $G$.
\end{definition}

\begin{proof}[Proof of Theorem~\ref{thm:polytimeTranslatedCenSymConvex}]
Raghavan and Spinrad have given a polynomial time algorithm that takes an abstract graph as input, and returns a CNEEO or a certificate that no CNEEO exists for the graph. Secondly, they showed how to compute in polynomial time a maximum clique when given a graph and a CNEEO on it. Therefore, it is sufficient to show that for any centrally symmetric, bounded, convex set $S$, and any intersection graph $G$ of translated of $S$, there exists a CNEEO on $G$. Let us consider such a graph $G$ with a representation. Arguing with Lemma~\ref{lem:D1clique} as previously, ordering the edges by non-increasing length gives a CNEEO, where the length of an edge is the distance between the two centers.
\end{proof}

\section{Homothets of a centrally symmetric convex set}\label{sec:homothets}

Here we observe that the EPTAS for \cli in disk graphs extends to the intersection graphs of homothets of a centrally symmetric convex set.
Bonamy et al. show:
\begin{theorem}[\cite{Bonamy18}]\label{thm:eptas}
  For any constants $d \in \mathbb N$, $0 < \beta \leqslant 1$, for every $0 < \varepsilon < 1$, there is a randomized $(1-\varepsilon)$-approximation algorithm running in time $2^{\tilde{O}({1/\varepsilon}^3)}n^{O(1)}$, and a deterministic PTAS running in time $n^{\tilde{O}({1/\varepsilon}^3)}$ for \cli on $n$-vertex graphs $G$ satisfying the following conditions:
  \begin{itemize}
\item there are no two mutually induced odd cycles in $\overline{G}$ (the complement of $G$),
\item the VC-dimension of the neighborhood hypergraph $\{N[v]~|~v \in V(G)\}$ is at most $d$, and
\item $G$ has a clique of size at least $\beta n$.
\end{itemize}
\end{theorem}

The first item is enough to obtain a subexponential-algorithm \cite{Bonnet18} and boils down to proving a structural lemma on the representation of $K_{2,2}$ (see \cref{lem:K22}).
We show that the previous theorem applies to more general shapes than disks.

\begin{theorem}\label{thm:homothets}
\cli admits a subexponential-time algorithm and an EPTAS in intersection graphs of homothets of a fixed bounded centrally symmetric convex set $S$.
\end{theorem}
We use the associated norm as defined in~\cref{sec:translates}, and check the three above conditions.


\begin{lemma}\label{lem:K22}
In a representation of $K_{2,2}$ with homothets of $S$ placing the four centers in convex position, the non-edges are between vertices corresponding to opposite corners of the quadrangle.
\end{lemma}

\begin{proof}
  \begin{figure}[tbhp]
	\centering
	\includegraphics{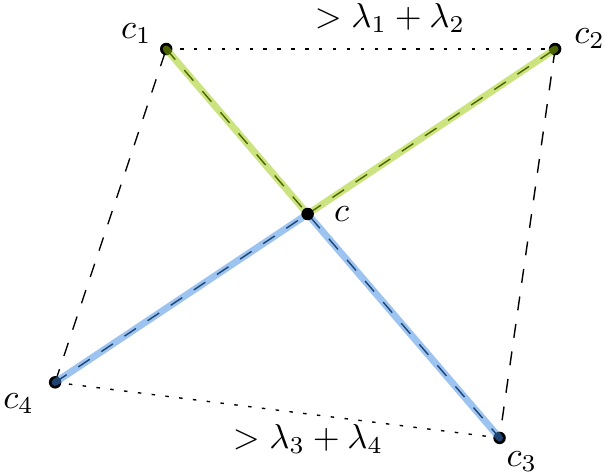}
	\caption{Illustration of the proof of Lemma~\ref{lem:K22}. 
	Non-edges are dotted and edges are dashed.
	By Assumption the top and bottom segment have a length of at least $\lambda_1+\lambda_2+\lambda_3+\lambda_4$. 
	By the triangle inequality the green plus the blue path are even longer.}
\end{figure}
  Let $S_1$, $S_2$, $S_3$ and $S_4$ be the four homothets.
  We denote by $c_i$ the center of $S_i$, and by $\lambda_i$ its scaling factor.
  Let us assume by contradiction that they appear in this order on the convex hull, that $S_1$ and $S_2$ make one non-edge, and that $S_3$ and $S_4$ make the other.
  By assumption, we have $\|c_1-c_2\|>\lambda_1+\lambda_2$, and likewise $\|c_3-c_4\|>\lambda_3+\lambda_4$.
  Let us denote by $c$ the intersection of the lines 
  $\ell(c_1,c_3)$ and $\ell(c_2,c_4)$.
  We have $\|c_1-c\|+\|c-c_2\|>\|c_1-c_2\|$ by triangular inequality.
  Likewise it holds $\|c_3-c\|+\|c-c_4\|>\|c_3-c_4\|$.
  We therefore obtain $\lambda_1+\lambda_2+\lambda_3+\lambda_4<\|c_1-c\|+\|c-c_2\|+\|c_3-c\|+\|c-c_4\|=\|c_1-c_3\|+\|c_2-c_4\|\leqslant \lambda_1+\lambda_3+\lambda_2+\lambda_4$, which is a contradiction.
\end{proof}

\cref{lem:K22} implies by some parity arguments that the first condition of~\cref{thm:eptas} holds (see Theorem 6 in \cite{Bonnet18}).
It is well known that a family of homothets forms a pseudo-disk arrangement.
Therefore the second property holds as shown by Aronov et al.~\cite{Aronov18}.
Finally we enforce the third condition of \cref{thm:eptas}, by using a chi-boundedness result of Kim et al.~\cite{Kim04}.
\begin{lemma}\label{lem:thirdCondition}
  With a polynomial multiplicative factor in the running time, one can reduce to instances satisfying the third condition of~\cref{thm:eptas} with $\beta=1/36$.   
\end{lemma}

\begin{proof}
  Kim et al.~\cite{Kim04} show that in any representation of an intersection graph $G$ of homothets of a convex set, a homothet $S$ with a smallest scaling factor has degree at most $6 \omega(G)-7$, where $\omega(G)$ denotes the clique number of $G$.
  Their proof also implies that the independence number of its neighborhood is at most $6$.
  By degenerence, the coloring number, denoted by $\chi(G)$ is at most $6\omega(G) - 6$.
  First we find in polynomial-time a vertex $v$ such that the independence number of its neighborhood is at most $6$.
  Let us denote by $G_v$ the subgraph induced by its neighborhood, and $n$ denotes its number of vertices.
  We denote by $\alpha(.)$ the independence number of a graph.
  As $G_v$ has a representation with homothets of $S$, we have $\chi(G_v)\leqslant 6 \omega(G_v)$.
  Therefore $\alpha(G_v) \omega(G_v) \geqslant \frac{1}{6} \alpha(G_v) \chi(G_v) \geqslant  \frac{1}{6}n$.
  Thus by assumption we have $\omega(G_v) \geqslant \frac{1}{36}n$.
  Then we can compute a maximum clique that contains $v$, or remove $v$ from the graph and iterate.
  The EPTAS of Bonamy et al. is called linearly many times.
\end{proof}

\end{document}